\documentclass{article}
\usepackage{graphicx}
\usepackage{amsmath}
\usepackage{amsthm} 
\usepackage{amssymb}
\usepackage{enumerate}
\usepackage{fullpage}
	
\usepackage{hyperref}
\usepackage{url,xspace}
\usepackage{mathbbol}
\usepackage[all]{xy}
\SelectTips{cm}{}

\theoremstyle{plain} 
\newtheorem{thm}{Theorem}[section]
\newtheorem{prop}[thm]{Proposition}
\newtheorem{lem}[thm]{Lemma}
\theoremstyle{definition}
\newtheorem{definition}[thm]{Definition}
\theoremstyle{remark}
\newtheorem{rem}[thm]{Remark}

\newcommand{\To}{\Rightarrow}
\newcommand{\id}{\mathit{id}} 
\newcommand{\Sig}{\mathrm{Sig}} 
\newcommand{\Ss}{\Sigma} 
\newcommand{\Tt}{\Theta} 
\newcommand{\logL}{\mathcal{L}} 
\newcommand{\pu}{\mathrm{pure}} 	
\newcommand{\ex}{\mathrm{exc}} 	
\newcommand{\pure}{{(0)}} 
\newcommand{\ppg}{{(1)}} 
\newcommand{\ctc}{{(2)}} 
\newcommand{\eqs}{\equiv} 
\newcommand{\eqw}{\sim} 
\newcommand{\empt}{\mathbb{0}}  
\newcommand{\Exc}{\mathit{Exc}} 
\newcommand{\inn}{\mathit{normal}} 
\newcommand{\ina}{\mathit{abrupt}} 
\newcommand{\cotuple}[1]{\left[ #1 \right]}
\newcommand{\bigcotuple}[1]{\bigl[ #1 \bigr]}
\newcommand{\Bigcotuple}[1]{\Bigl[ #1 \Bigr]}

\newcommand{\cotu}{\cotuple{\,}}
\newcommand{\toppg}[1]{\triangledown #1 } 
\newcommand{\bigtoppg}[1]{\bigtriangledown #1 } 
\newcommand{\throw}[2]{\mathit{throw}_{#2}\,E_{#1}\,} 
\newcommand{\try}[2]{\mathit{try}\{#1\}\,#2}  
\newcommand{\catch}[2]{\mathit{catch}\,\{E_{#1}\To #2\}}  
\newcommand{\catchn}[4]{\mathit{catch}\,\{E_{#1}\!\To\! #2|\dots|E_{#3}\!\To\! #4\}}  
\newcommand{\catchij}[4]{\mathit{catch}\,\{E_{#1}\!\To\! #2\;|\;E_{#3}\!\To\! #4\}}  

\title{Decorated proofs for computational effects: Exceptions} 
\date{March 13., 2012}
\author{
  Jean-Guillaume Dumas\thanks{
    LJK, Universit\'e de Grenoble, France. \url{Jean-Guillaume.Dumas@imag.fr}.
  This work is partly funded by the project HPAC 
  of the French Agence Nationale de la Recherche (ANR 11 BS02 013).}, 
  Dominique Duval\thanks{
    LJK, Universit\'e de Grenoble, France. \url{Dominique.Duval@imag.fr}.
  This work is partly funded by the project CLIMT 
   of the French Agence Nationale de la Recherche (ANR 11 BS02 016).}, 
  Laurent Fousse\thanks{
    LJK, Universit\'e de Grenoble, France. \url{Laurent.Fousse@imag.fr}}, 
  Jean-Claude Reynaud\thanks{
    Malhivert, Claix, France. \url{Jean-Claude.Reynaud@imag.fr}}.
}

\begin{document}

\maketitle

\begin{abstract} 
We define a proof system for exceptions which is close to 
the syntax for exceptions, in the sense that the exceptions 
do not appear explicitly in the type of any expression. 
This proof system is sound with respect to the intended 
denotational semantics of exceptions. With this inference system 
we prove several properties of exceptions.
Semantics of exceptions.
Proof system.
\end{abstract}

\section*{Introduction}

In this paper, as in the apparented papers 
\cite{DD10-dialog,DDR11-seqprod,DDFR12b-state}, 
we consider that a \emph{computational effect} in a language 
corresponds to an apparent lack of soundness: 
the intended denotational semantics is not a model of the syntax, 
but it becomes so when the syntax is endowed with relevant \emph{decorations}; 
more precisely, a proof system can be designed for dealing 
with these decorations, which is sound with respect to 
the intended denotational semantics. 
In \cite{DDFR12b-state} this point of view has been applied to the 
side-effects due to the evolution of the \emph{states} of the memory
in an imperative or object-oriented language.
In this paper, it is applied to the effects caused by \emph{exceptions}. 
It happens that there is a duality between the denotational semantics of states
and the \emph{core part} of the semantics of exceptions \cite{DDFR12a-short}. 
The encapsulation of the core part inside the mechanism of exceptions is 
a succession of case distinctions;
the proof system is extended for dealing with it. 
Properties of exceptions can be proved using this inference system 
and the proofs can be simplified by re-using proofs on states, 
thanks to the duality.

To our knowledge, the first categorical treatment of 
computational effects is due to Moggi \cite{Mo91};
this approach relies on monads, it is implemented 
in the programming language Haskell \cite{Wa92,haskell}. 
Although monads are not used in this paper, 
the basic ideas underlying our approach rely on Moggi's remarks 
about notions of computations and monads. 
The examples proposed by Moggi include 
the exceptions monad $T A = A + E$ where $E$ is the set of exceptions. 
Later on, using the correspondence between monads
and algebraic theories, Plotkin and Power proposed to 
use Lawvere theories for dealing with 
the operations and equations related to computational effects  
\cite{PP02,HP07};  
an operation is called \emph{algebraic} 
when it satisfies some relevant genericity properties. 
The operation for raising exceptions is algebraic, 
while the operation for handling exceptions is not \cite{PP03}. 
It follows that the handling of exceptions is quite difficult 
to formalize in this framework;
several solutions are proposed in \cite{SM04,Le06,PP09}.
In this paper we rather use the categorical approach 
of \emph{diagrammatic logics}, as introduced in \cite{D03-diaspec}
and developed in \cite{DD10-dialog}. 

In Section~\ref{sec:exc} a denotational semantics for exceptions is 
defined, where we dissociate the core operations from their encapsulation. 
Then a decorated proof system and a decorated specification 
for exceptions 
are defined in Section~\ref{sec:deco} and it is checked that 
the denotational semantics for exceptions can be seen as a model 
of this specification. 
In Section~\ref{sec:proof} we use this framework for
proving some properties of exceptions. 

\section{Denotational semantics for exceptions}
\label{sec:exc}

In this Section we define a denotational semantics 
of exceptions which relies on the semantics 
of exceptions in various languages, for instance 
in Java \cite{java} and ML \cite{sml}. 
Syntax is introduced in Section~\ref{subsec:exc-syntax} 
and the distinction between ordinary and exceptional values is 
discussed in Section~\ref{subsec:exc-sem}.  
Denotational semantics of raising and handling exceptions are 
considered in Sections~\ref{subsec:exc-tag} and~\ref{subsec:exc-untag}, 
respectively. 

\subsection{Signature for exceptions}
\label{subsec:exc-syntax}

The syntax for exceptions in computer languages depends on the language:  
the keywords for raising exceptions may be either 
\texttt{raise} or \texttt{throw}, 
and for handling exceptions they may be either 
\texttt{handle} or \texttt{try-catch}, for instance. 
In this paper we rather use \texttt{throw} and \texttt{try-catch}, 
but this choice does not matter. 
More precisely, 
the syntax of our language may be described in two parts: 
a \emph{pure} part and an \emph{exceptional} part. 
The pure part is a signature $\Sig_{\pu}$, 
made of types and operations; 
the $\Sig_{\pu}$-expressions are called the \emph{pure expressions}.  
The interpretation of the pure expressions should 
neither raise nor handle exceptions. 
We assume that the pure operations are either constants or unary.
General $n$-ary operations would require the use of 
sequential products, as in \cite{DDR11-seqprod}; 
in order to focus on the fundamental properties of exceptions 
they are not considered in this paper. 
The exceptional part is made of a symbol $E_i$
for each index $i$ in some set of indices $I$, 
which is declared as:
  $ \mathit{Exception}\; E_i \;\mathit{of}\; P_i $, 
where $P_i$ is a pure type  called the \emph{type of parameters} 
for the \emph{exceptional type} $E_i$ (the $P_i$'s need not be distinct). 
The exceptional types $E_i$ provide familiar notations 
for the raising and handling operations  
and in Section~\ref{subsec:exc-tag} they are interpreted as sets,
however  we will not define any expression of type $E_i$.

Let us assume that the signature $\Sig_{\pu}$ is fixed. 
The \emph{expressions} of our language are defined recursively from 
the pure operations and from the \emph{raising}
and \emph{handling} operations, as follows. 

\begin{definition}
\label{defi:exc-sig}
Given a set of indices $I$ and a symbol $E_i$ for each $i\in I$, 
the \emph{signature for exceptions} $\Sig_{\ex}$ is made of $\Sig_{\pu}$
together with  
a \emph{raising} operation
for each $i$ in $I$ and each type $Y$ in $\Sig_{\pu}$:
  $$ \throw{i}{Y} :P_i\to Y \;.$$ 
and a \emph{handling} operation 
for each $\Sig_{\ex}$-expression $f:X\to Y$, 
each non-empty list of indices $(i_1,\dots,i_n)$ 
and each $\Sig_{\ex}$-expressions $g_1:P_{i_1}\to Y$, \dots, $g_n:P_{i_n}\to Y$: 
  $$ \try{f}{\catchn{i_1}{g_1}{i_n}{g_n}} : X \to Y \;.$$ 
\end{definition}

\subsection{Ordinary values and exceptional values} 
\label{subsec:exc-sem}

The syntax for exceptions defined in Section~\ref{subsec:exc-syntax} 
is now interpreted in the category of sets. 
In order to express the denotational semantics of exceptions,
a major point is that there are two kinds of values:
the ordinary (or non-exceptional) values and the exceptions.
It follows that the operations may be classified 
according to the way they may, or may not, interchange 
these two kinds of values: 
an ordinary value may be \emph{tagged} for constructing an exception, 
and later on the tag may be cleared in order to recover the value.
Then we say that the exception gets \emph{untagged}. 
Let us introduce a set $\Exc$ called the \emph{set of exceptions}.
For each set $X$ we consider the disjoint union $X+\Exc$ 
with the inclusions $\inn_X: X \to X+\Exc$ and $\ina_X: \Exc \to X+\Exc$.

\begin{definition}
\label{defi:exc-fcts}
For each set $X$, an element of $X+\Exc$ is 
an \emph{ordinary value} if it is in $\inn_X(X)$ 
and an \emph{exceptional value} if it is in $\ina_X(\Exc)$.
A function $f:X+\Exc \to Y+\Exc$ is said to  
\emph{raise an exception} 
if there is an element $x\in X$ such that $f(x)\in\Exc$;
\emph{propagate exceptions} 
if $f(\ina_X(e))=\ina_Y(e)$ for every $e\in\Exc$;  
\emph{recover from an exception} 
if there is some $e\in\Exc$ such that $f(e)\in Y$.
\end{definition}

We will use the same notations for the syntax 
and for its interpretation. 
Each type $X$ is interpreted as a set~$X$.
Each pure expression $f_0:X\to Y$ is interpreted as a function $f_0:X\to Y$, 
which can be extended as $f=\inn_Y\circ f_0:X \to Y+\Exc$.
When $f:X\to Y$ is a $\Sig_{\ex}$-expression, which 
may involve some raising or handling operation, 
its interpretation is a function $f:X \to Y+\Exc$
which is defined in the next Sections~\ref{subsec:exc-tag} 
and~\ref{subsec:exc-untag}. 
In addition, every function $f:X \to Y+\Exc$ 
can be extended as $\cotuple{f\,|\,\ina_Y} :X+\Exc \to Y+\Exc$,
which is defined by the equalities 
$ \cotuple{f\,|\,\ina_Y} \circ \inn_X  = f$ and 
$ \cotuple{f\,|\,\ina_Y} \circ \ina_X  = \ina_Y$.
This is the unique extension of $f$ to $X+\Exc$ which propagates exceptions. 

\begin{rem}
\label{rem:ppg}  
\emph{The interpretation of a $\Sig_{\ex}$-expression $f:X\to Y$
is a function which propagates exceptions; 
this function may raise exceptions but it cannot 
recover from an exception.} 
In Section~\ref{subsec:exc-untag}, in order to catch exceptions, 
we will introduce functions which recover from exceptions. 
However such a function cannot be the interpretation of 
any $\Sig_{\ex}$-expression. 
Indeed, a \emph{try-catch} expression may recover from exceptions  
which are raised inside the \emph{try} block, 
but if an exception is raised before the \emph{try-catch} expression
is evaluated, this exception is propagated. 
Recovering from an exception can only be done by functions 
which are not expressible in the language generated by $\Sig_{\ex}$: 
such functions are called the \emph{untagging} functions,
they are defined in Section~\ref{subsec:exc-untag}.
Together with the \emph{tagging} functions 
defined in Section~\ref{subsec:exc-tag} they are called 
the \emph{core} functions for exceptions.
\end{rem}

\subsection{Tagging and raising exceptions: \textit{throw}}
\label{subsec:exc-tag}

Raising an exception is based on a tagging process,
modelled as follows. 

\begin{definition}
\label{defi:exc-tag}
For each index $i\in I$ there is an injective function
  $t_i:P_i\to\Exc$, 
called the \emph{exception constructor} or the \emph{tagging} function 
of index $i$, 
and the tagging functions for distinct indices have disjoint images. 
The image of $t_i$ in $\Exc$ is denoted $E_i$. 
\end{definition}

Thus, the tagging function $t_i:P_i\to\Exc$
maps a non-exceptional value (or \emph{parameter}) $a\in P_i$ 
to an exception $t_i(a)\in\Exc$. 
This means that the non-exceptional value $a$ in $P_i$ gets 
tagged as an exception $t_i(a)$ in $\Exc$. 
The disjoint union of the $E_i$'s is a subset of $\Exc$; 
for simplicity we assume that 
  $ \Exc = \sum_{i\in I} E_i \;.$ 

\begin{definition}
\label{defi:exc-raise}
For each index $i\in I$ and each set $Y$, the \emph{throwing} 
or \emph{raising} function $ \throw{i}{Y} $ 
is the tagging function $t_i$ followed by 
the inclusion of $\Exc$ in $Y+\Exc$: 
  $ \throw{i}{Y} = \ina_Y \circ t_i : P_i\to Y+\Exc \;.$ 
\end{definition}

\subsection{Untagging and handling exceptions: \textit{try-catch} }
\label{subsec:exc-untag}

Handling an exception is based on an untagging process 
for clearing the exception tags, which is modelled as follows. 

\begin{definition}
\label{defi:exc-untag}
For each index $i\in I$ there is a function
  $c_i:\Exc\to P_i+\Exc$,   
called the \emph{exception recovery} or the \emph{untagging} function
of index $i$, 
which satisfies: 
$ \forall a\in P_i\; c_i(t_i(a)) = a \mbox{ and }  
  \forall b\in P_j\; c_i(t_j(b)) = t_j(b) \mbox{ for each } j\ne i \;.$
\end{definition}

Thus, for each $e\in\Exc$ 
the untagging function $c_i(e)$ 
tests whether the given exception $e$ is in $E_i$; 
if this is the case, 
then it returns the parameter $a\in P_i$ such that $e=t_i(a)$, 
otherwise it propagates the exception~$e$. 
Since it has been assumed that $\Exc=\sum_{j\in I} E_j$, 
the untagging function $c_i(e)$ is uniquely determined by 
the above equalities.
 
For handling exceptions of type $E_i$ raised by some function $f:X\to Y+\Exc$, 
for $i$ in a non-empty list $(i_1,\dots,i_n)$ of indices,
one provides for each $k$ in $\{1,\dots,n\}$ 
a function $g_k:P_{i_k}\to Y+\Exc$ 
(which thus may itself raise exceptions). 
Then the handling process builds a function 
which encapsulates some untagging functions 
and which propagates exceptions. 
The indices $i_1,\dots,i_n$ form a list: they are given in this order 
and they need not be pairwise distinct. 
It is assumed that this list is non-empty,
because it is the usual choice in programming languages,
however it would be easy to drop this assumption. 

\begin{definition}
\label{defi:exc-handle}
For each function $f:X\to Y+\Exc$,
each non-empty list $(i_1,\dots,i_n)$ of indices in $I$ 
and each family of functions $g_k:P_{i_k}\to Y+\Exc$ 
(for $k\in\{1,\dots,n\}$), 
the \emph{handling} function 
  $$ \try{f}{\catchn{i_1}{g_1}{i_n}{g_n}} : X\to Y+\Exc $$  
is defined as follows.
Let $h=\try{f}{\catchn{i_1}{g_1}{i_n}{g_n}}$, for short. 
For each $x\in X$, $\;h(x) \in Y+\Exc$ is defined in the following way.
\begin{itemize}
\item[] \textit{First $f(x)$ is computed: }
\\ let $y=f(x) \in Y+\Exc$. 
  \begin{enumerate}[(1)]
  \item \textit{If $y$ is not an exception, then it is the required result: } 
  \\ if $y\in Y$ then $h(x) = y\in Y \subseteq Y+\Exc$.
  \item \textit{If $y$ is an exception, then: } 
    \begin{enumerate}[(a)]
    \item \textit{If the type of $y$ is $E_i$ for some index $i$ 
    in $(i_1,\dots,i_n)$, then $y$ has to be caught 
    according to the first occurrence of the index $i$ in the list: } 
    \\ for each $k=1,\dots,n$,
      \begin{itemize}
      \item \textit{Check whether the exception $y$ has type $E_{i_k}$: } 
      \\ let $z = c_{i_k}(y) \in P_{i_k}+\Exc$. 
      \item \textit{If the exception $y$ has type $E_{i_k}$ then it is caught: }
      \\ if $z\in P_{i_k}$ then $h(x) = g_k(z)\in Y+\Exc$.
      \end{itemize}
    \item \textit{If the type of $y$ is $E_i$ for some $i\not\in
    \{i_1,\dots,i_n\}$, then $y$ is propagated: } 
    \\ otherwise $h(x) = y\in \Exc \subseteq Y+\Exc$.
    \end{enumerate}
  \end{enumerate}
\end{itemize}
\end{definition}

Equivalently, the definition of $h= \try{f}{\catchn{i_1}{g_1}{i_n}{g_n}} $ 
can be expressed as follows. 
\begin{description}
\item[(1-2)] The function $h:X\to Y+\Exc$ is defined 
from $f$ and from a function 
$\catchn{i_1}{g_1}{i_n}{g_n} : \Exc\to Y+\Exc$ by: 
  \begin{equation}
  \label{diag:handle-explicit-intermediate} 
  h \;=\; 
    \Bigcotuple{\; \inn_Y \;|\; \catchn{i_1}{g_1}{i_n}{g_n} \;} \circ f 
  \end{equation} 
  $$ \xymatrix@C=3pc@R=1pc{
  & Y \ar[d]_{\inn} \ar[rrrrd]^{\inn} &&&& \\
  X \ar[r]^{f} & 
    Y+\Exc \ar[rrrr]^(.4){h} &
    \ar@{}[ul]|(.4){=} \ar@{}[dl]|(.4){=} &&& Y+\Exc \\ 
  & \Exc \ar[u]^{\ina} \ar[rrrru]_{\catchn{i_1}{g_1}{i_n}{g_n}} &&&& \\
  } $$
\begin{description}
\item[(a-b)] The function $\catchn{i_1}{g_1}{i_n}{g_n}$ 
is obtained by setting $p=1$ in the family of functions 
$k_p = \catchn{i_p}{g_p}{i_n}{g_n} : \Exc\to Y+\Exc$
(for $p=1,\dots,n$) which are defined recursively by: 
  \begin{equation}
  \label{diag:handle-explicit-catch} 
  k_p \;=\;  
    \begin{cases} 
     \bigcotuple{\; g_n \;|\; \ina_Y \;} \circ c_{i_n} & 
		   \mbox{ when } p=n \\
     \bigcotuple{\; g_p \;|\; k_{p+1} \;}  \circ c_{i_p} & 
		   \mbox{ when } p<n \\
   \end{cases} 
  \end{equation}
Let $k_{n+1}=\ina_Y$, then
 $ k_p \;=\; \cotuple{\; g_p \;|\; k_{p+1} \;} \circ c_{i_p}
      \;\;  \mbox{ for each } p\leq n $.
 $$ \xymatrix@C=3pc@R=1pc{
  & P_{i_p} \ar[d]_{\inn} \ar[rrrrd]^{g_p} &&&& \\
  \Exc \ar[r]^(.4){c_{i_p}} & 
    P_{i_p}+\Exc 
      \ar[rrrr]^(.4){\cotuple{g_p|k_{p+1}}} &
    \ar@{}[ul]|(.4){=} \ar@{}[dl]|(.4){=} &&& Y+\Exc \\ 
  & \Exc \ar[u]^{\ina} \ar[rrrru]_{k_{p+1}} &&&& \\
  }  $$
\end{description}
\end{description}
When $n=1$ we get $ 
\try{f}{\catch{i}{g}} =  
  \Bigcotuple{\inn_Y |
       \bigcotuple{g| \ina_Y  
               } \circ c_i 
            } \circ f  \;.$

\begin{rem}
\label{rem:cases}
The handling process involves several nested case distinctions.
Since it propagates exceptions, there is a first case distinction for 
checking whether the argument $x$ is an exception (which is simply propagated) 
or not. If $x$ is not an exception, then there is 
a case distinction (1-2) for checking whether $f(x)$ is an exception or not.
If $f(x)$ is an exception 
then each step (a-b) checks whether the result of the 
untagging function is an exception.
All these case distinctions 
check whether some value is an exception or not,
they rely on disjoint unions of the form $T+\Exc$. 
In contrast, for each step (a-b) 
there is another case distinction encapsulated in the computation of 
the untagging function, which checks whether the exception 
has the required exception type and relies on the disjoint union 
$\Exc=\sum_i E_i$.  
\end{rem}

\section{Decorated logic for exceptions}
\label{sec:deco}

In Section~\ref{sec:exc} we have introduced 
a signature $\Sig_{\ex}$ and a denotational semantics for exceptions.
However the soundness property is not satisfied:
the denotational semantics is not a model of the signature,
in the usual sense, 
since an expression $f:X \to Y$ is interpreted as a function $f:X \to Y+\Exc$  
instead of $f:X\to Y$. 
Therefore, in this Section we build a 
\emph{decorated specification} for exceptions,
including a ``decorated'' signature and ``decorated'' equations, 
which is sound with respect to the denotational semantics 
of Section~\ref{sec:exc}. 
For this purpose, first we form an equational specification by 
extending the signature $\Sig_{\ex}$ with operations $t_i$ and $c_i$ 
and equations involving them,  
in order to formalize the tagging and untagging functions 
of Sections~\ref{subsec:exc-tag} and~\ref{subsec:exc-untag}. 
Then we add \emph{decorations} to this specification, 
and we define the interpretation of the 
expressions and equations according to their decorations. 
This means that we have to extend the equational logic with 
a notion of decoration; the decorations and  
the decorated inference rules
are given in Section~\ref{subsec:deco-deco}.
In Section~\ref{subsec:deco-exc} we define the 
decorated specification for exceptions 
and in  Section~\ref{subsec:deco-model} we check 
that this decorated specification
is sound with
respect to the denotational semantics of Section~\ref{sec:exc}.
In the decorated specification for exceptions, 
there are on one side 
\emph{private} operations for tagging and untagging exceptions, 
which do not appear in the signature for exceptions $\Sig_{\ex}$,
and on the other side  
\emph{public} operations for raising and handling exceptions, 
which are defined using the private operations.  
According to remark~\ref{rem:ppg}, 
an important feature of exceptions is that 
\emph{all public operations propagate exceptions},
such operations will be called \emph{propagators}; 
operations for recovering from exceptions 
may appear only as private operations, 
which will be called \emph{catchers}.

\subsection{Decorations}
\label{subsec:deco-deco}

In order to deal with exceptions 
we define three decorations for expressions.
They are denoted by $\pure$, $\ppg$ and $\ctc$ used as superscripts,
and their meaning is described in an informal way as follows. 
\begin{itemize}
\item The interpretation of a \emph{pure} expression $f^\pure$ 
may neither raise exceptions nor recover form exceptions. 
\item The interpretation of a \emph{propagator} $f^\ppg$ 
may raise exceptions but it is not allowed to 
recover from exceptions; thus, it must propagate all exceptions. 
\item  The interpretation of a \emph{catcher} $f^\ctc$ 
may raise exceptions and recover form exceptions. 
\end{itemize}
Every pure expression can be seen as a propagator 
and every propagator as a catcher.  
It follows that every expression can be seen as a catcher, 
so that the decoration $\ctc$ could be avoided;  
however we often use it for clarity. 

In addition, we define two decorations for equations. 
They are denoted by two distinct relational symbols $\eqs$ 
for \emph{strong} equations and by $\eqw$ for  \emph{weak} equations. 
Using the fact that every expression can be seen as a catcher, 
their meaning can be described as follows. 
\begin{itemize}
\item A \emph{strong} equation $f^\ctc\eqs g^\ctc$ 
is interpreted as an equality of the functions $f$ and $g$
both on ordinary and on exceptional values. 
\item A \emph{weak} equation  $f^\ctc\eqw g^\ctc$ 
is interpreted as an equality of the functions $f$ and $g$
on ordinary values, but $f$ and $g$ may differ on exceptional values. 
\end{itemize}
Clearly every strong equation $f\eqs g$ gives rise 
to the weak equation $f\eqw g$. 
On the other hand, since propagators 
cannot modifiy the exceptional values, 
every weak equation between propagators 
can be seen as a strong equation, 
and a similar remark holds for pure expressions.

\begin{rem}
\label{rem:toppg}
It follows from these descriptions that 
every catcher $k$ gives rise to a propagator $\toppg{k}$ 
with a weak equation $k\eqw \toppg{k}$: 
this propagator $\toppg{k}$ has the same interpretation as $k$ 
on the non-exceptional values and it is interpreted as 
the identity on the exceptional values. 
\end{rem}

In  the short note \cite{DDFR12a-short} it is checked that,
from a denotational point of view, 
the functions for tagging and untagging exceptions are respectively 
\emph{dual},
in the categorical sense, 
to the functions for looking up and updating states.  
It happens that this duality also holds from the decorated 
point of view. 
Thus, most of the decorated rules for exceptions
are dual to the decorated rules for states \cite{DDFR12b-state}. 
The decorated rules for exceptions are given here 
in three parts (Figures~\ref{fig:proof-rules-one}, 
\ref{fig:proof-rules-two} and~\ref{fig:proof-rules-three}). 
For readability, the decoration properties are often grouped 
with other properties: for instance, ``$f^\ppg \eqw g^\ppg$''
means ``$f^\ppg$ and $g^\ppg$ and $f \eqw g$''. 

The rules in Figure~\ref{fig:proof-rules-one} may be called 
the rules for the 
\emph{decorated monadic equational logic} for exceptions.   
The unique difference between 
these rules and the dual rules for states 
lies in the congruence rules for the weak equations:
for states the replacement rule is restricted to pure $g$'s,
while for exceptions it is 
the substitution rule which is restricted to pure $f$'s. 

\begin{figure}[!ht]
\renewcommand{\arraystretch}{2.3}  
$$ \begin{array}{|c|} 
\hline 
\dfrac{X}{\id_X:X\to X } \qquad 
\dfrac{X}{\id_X^\pure } \qquad 
\dfrac{f^\pure}{f^\ppg} \\
\dfrac{f:X\to Y \quad g:Y\to Z}{g\circ f:X\to Z}  \qquad
\dfrac{f^\pure \quad g^\pure}{(g\circ f)^\pure}  \qquad 
\dfrac{f^\ppg \quad g^\ppg}{(g\circ f)^\ppg} \\  
\dfrac{f:X\to Y}{f\circ \id_X \eqs f} \qquad 
\dfrac{f:X\to Y}{\id_Y\circ f \eqs f} \qquad
\dfrac{f:X\to Y \quad g:Y\to Z \quad h:Z\to W}
  {h\circ (g\circ f) \eqs (h\circ g)\circ f} \\
\dfrac{f^\ppg \eqw g^\ppg}{f \eqs g} \qquad
\dfrac{f \eqs g}{f \eqw g} \\ 
\dfrac{}{f \eqs f} \qquad 
\dfrac{f \eqs g}{g \eqs f} \qquad 
\dfrac{f \eqs g \quad g \eqs h}{f \eqs h} \\ 
\dfrac{}{f \eqw f} \qquad 
\dfrac{f \eqw g}{g \eqw f} \qquad 
\dfrac{f \eqw g \quad g \eqw h}{f \eqw h} \\ 
\dfrac{f:X\to Y \quad g_1\eqs g_2:Y\to Z}
  {g_1\circ f \eqs g_2\circ f }  \qquad 
\dfrac{f_1\eqs f_2:X\to Y \quad g:Y\to Z}
  {g\circ f_1 \eqs g\circ f_2 } \\
\dfrac{f^\pure:X\to Y \quad g_1\eqw g_2:Y\to Z}
  {g_1\circ f \eqw g_2\circ f }  \qquad 
\dfrac{f_1\eqw f_2:X\to Y \quad g:Y\to Z}
  {g\circ f_1 \eqw g\circ f_2 } \\
\hline 
\end{array}$$
\renewcommand{\arraystretch}{1}
\caption{Decorated rules for exceptions (1)} 
\label{fig:proof-rules-one} 
\end{figure}


Several kinds of decorated coproducts are used for dealing with exceptions.
The rules in Figure~\ref{fig:proof-rules-two} are the 
rules for a decorated initial type~$\empt$, also called an \emph{empty type},  
and for a \emph{constitutive coproduct}, as defined below. 
These rules are dual to the rules for the decorated final type 
and for the observational product for states 
in \cite{DDFR12b-state}. 

\begin{definition}
\label{defi:initial}
A \emph{decorated initial type} for exceptions 
is a type $\empt$ such that  
for every type $X$ 
there is a pure expression $\cotu_X:\empt\to X$ 
such that every function from $\empt$ to $X$ is weakly equivalent to $\cotu_X$. 
\end{definition}

It follows that every pure expression and every propagator 
from $\empt$ to $X$ is strongly equivalent to $\cotu_X$.  

\begin{definition}
\label{defi:constitutive-coproduct} 
A \emph{constitutive coproduct} for exceptions 
is a family of propagators 
$(q_i:X_i\to X)_i$ such that 
for every family of propagators $(f_i:X_i\to Y)_i$ 
there is a catcher $f=\cotuple{f_i}_i:X\to Y$,
unique up to strong equations,  
such that $f \circ q_i \eqw f_i $ for each $i$. 
\end{definition}

This definition means that a constitutive coproduct 
can be used for building a catcher from several propagators; 
this corresponds to the fact that the set $\Exc$ is 
the disjoint union of the $E_i$'s. 

\begin{figure}[!ht]
\renewcommand{\arraystretch}{2.3}  
$$ \begin{array}{|c|} 
\hline 
\multicolumn{1}{|l|}
  {\qquad \text{When $\empt$ is a \emph{decorated initial type:}} } \\ 
\dfrac{X}{\cotu_X:\empt \to X} \qquad  
\dfrac{X}{\cotu_X^\pure} \qquad 
\dfrac{f:\empt \to Y}{f \eqw \cotu_Y} \\
\multicolumn{1}{|l|}
  {\qquad 
  \text{When $(q_i^\ppg:X_i\to X)_i$ is a \emph{constitutive coproduct:}} 
  \qquad } \\ 
\dfrac {(f_i^\ppg : X_i\to Y)_i} 
  {\cotuple{f_i}_i^\ctc : X\to Y} \qquad   
\dfrac {(f_i^\ppg : X_i\to Y)_i} 
  {\cotuple{f_j}_j \circ q_i \eqw f_i } \\ 
\dfrac{(f_i^\ppg : X_i\to Y)_i \quad f^\ctc : X\to Y \quad 
  \forall i \; f \circ q_i \eqw f_i}
  {f \eqs \cotuple{f_j}_j }  \\
\hline 
\end{array}$$
\renewcommand{\arraystretch}{1}
\caption{Decorated rules for exceptions (2)} 
\label{fig:proof-rules-two} 
\end{figure}

The next property corresponds to remark~\ref{rem:toppg}.

\begin{definition}
\label{defi:deco-exc-coprod-ppg}
For each catcher $k^\ctc:X\to Y$ there is a propagator $\toppg{k}^\ppg:X\to Y$,
unique up to strong equations, such that 
$ \toppg{k}^\ppg \eqw k^\ctc \;.$
\end{definition}

According to the previous rules, 
for each type $X$ there are two pure expressions 
$\id_X:X\to X$ and $\cotu_X:\empt\to X$. 
It is straightforward to check that they 
form a coproduct with respect to pure expressions
and strong equations: 
for each $f^\pure:X\to Y$ and $g^\pure:\empt\to Y$
there is a pure expression $\cotuple{f|g}^\pure:X\to Y$, 
unique up to strong equations, 
such that $\cotuple{f|g}\circ\id_X\eqs f$ and 
$\cotuple{f|g}\circ\cotu_X\eqs g$,
indeed such a situation implies that $g\eqs \cotu_X$
and $\cotuple{f|g} \eqs f$. 
This pure coproduct, with coprojections $\id_X$ and $\cotu_X$,
is called \emph{the coproduct $X\cong X+\empt$}. 
In addition, we assume that it satisfies the following \emph{decorated 
coproduct} property.  

\begin{definition}
\label{defi:deco-exc-coprod-ctc}
For each propagator $g^\ppg:X\to Y$ and each catcher $k^\ctc:\empt\to Y$ 
there is a catcher $\cotuple{g\,|\,k}^\ctc:X\to Y$,
unique up to strong equations, such that 
$ \cotuple{g\,|\,k}^\ctc \eqw g^\ppg $ and 
  $ \cotuple{g\,|\,k}^\ctc \circ \cotu_X^\pure \eqs k^\ctc \;. $
\end{definition}

The rules in Figure~\ref{fig:proof-rules-three} are the 
rules for the construction of $\toppg{k}$
and for the decorated coproduct $X\cong X+\empt$. 
They will be used for building the 
handling operations from the untagging operations. 

\begin{figure}[!ht]
\renewcommand{\arraystretch}{2.3}  
$$ \begin{array}{|c|} 
\hline 
\dfrac{k^\ctc:X\to Y}{\toppg{k}^\ppg:X\to Y} \qquad  
\dfrac{k^\ctc:X\to Y}{\toppg{k} \eqw k} \\
\;\dfrac{g^\ppg\!:\!X\!\to\! Y \;\; k^\ctc\!:\!\empt\!\to\! Y}
  {\cotuple{g\,|\,k}^\ctc\!:\!X\to Y}
\;\; 
\dfrac{g^\ppg\!:\!X\!\to\! Y \;\; k^\ctc\!:\!\empt\!\to\! Y}
  {\cotuple{g\,|\,k} \eqw g}
\;\;
\dfrac{g^\ppg\!:\!X\!\to\! Y \;\; k^\ctc\!:\!\empt\!\to\! Y}
  {\cotuple{g\,|\,k} \circ \cotu_X \eqs k } \; \\ 
\dfrac{g^\ppg : X\to Y \quad k^\ctc : \empt\to Y 
  \quad f^\ctc : X\to Y \quad  f \eqw g \quad 
  f \circ \cotu_X \eqs k}{f\eqs \cotuple{g\,|\,k}} \\
\hline 
\end{array}$$
\renewcommand{\arraystretch}{1}
\caption{Decorated rules for exceptions (3)} 
\label{fig:proof-rules-three} 
\end{figure}

\subsection{A decorated specification for exceptions}
\label{subsec:deco-exc}

Let $\logL$ denote the inference system provided by
the decorated rules for exceptions (Figures~\ref{fig:proof-rules-one}, 
\ref{fig:proof-rules-two} and~\ref{fig:proof-rules-three}). 
As for other inference systems, 
we may define \emph{theories} and \emph{specifications}
(or \emph{presentations of theories}) 
with respect to $\logL$.
They are called \emph{decorated specifications}
and \emph{decorated theories}, respectively. 
This approach is based on the general framework for 
\emph{diagrammatic} theories and specifications 
\cite{D03-diaspec,DD10-dialog}, 
but no knowledge of this framework is assumed in this paper. 
A decorated theory is made of 
types, expressions, equations and coproducts which satisfy 
the decorated rules for exceptions. 
In this Section we define a decorated specification $\Ss_{\ex}$,
which may be used for generating a decorated theory by applying
the decorated inference rules for exceptions. 
 
\begin{definition}
\label{defi:deco-spec-exc} 
Let $\Ss_{\pu}$ be some fixed equational specification
(as in Section~\ref{subsec:exc-syntax} for simplicity it is assumed that 
$\Ss_{\pu}$ has no $n$-ary operation with $n>1$). 
The \emph{decorated specification for exceptions} $\Ss_{\ex}$ 
is made of the equational specification $\Ss_{\pu}$
where each operation is decorated as pure and each equation as strong
(or weak, since both coincide here),
a type $\empt$ called the \emph{empty} type 
and for each $i$ in some set $I$  
a type $P_i$ (of \emph{parameters}) in $\Ss_{\pu}$, 
a propagator $t_i^{\ppg}:P_i\to\empt$, 
a catcher $c_i^{\ctc}: \empt \to P_i$  
and the weak equations:  
$  
   c_i \circ t_i \eqw \id : P_i \to P_i $
and 
$  
   c_i \circ t_j \eqw \cotu \circ t_j : P_j \to P_i  \;
     \mbox{ for every } j \in I,\; j\ne i $. 
\end{definition}

\begin{definition}
\label{defi:deco-raise} 
For each $i$ in $I$ and each type $Y$ in $\Ss_{\pu}$ 
the \emph{raising} propagator 
 $$ (\throw{i}{Y})^\ppg :P_i\to Y $$ 
is defined as 
  $$ \throw{i}{Y} = \cotu_Y \circ t_i : P_i\to Y \;.$$ 
\end{definition}

According to remark~\ref{rem:ppg}, the handling operation 
$\try{f}{\mathit{catch}\{\dots\} }$ 
is a propagator, not a catcher: indeed, it may recover from 
exceptions which are raised by $f$, 
but it must propagate exceptions which are raised before 
$\try{f}{\mathit{catch}\{\dots\}}$ is called. 

\begin{definition}
\label{defi:deco-handle} 
For each propagator $f^\ppg:X\to Y$, 
each non-empty list of indices $(i_1,\dots,i_n)$ 
and each propagators $g_1^\ppg:P_{i_1}\to Y,...,g_n^\ppg:P_{i_n}\to Y$, 
the \emph{handling} propagator 
  $$ (\try{f}{\catchn{i_1}{g_1}{i_n}{g_n}})^\ppg  : X \to Y $$ 
is defined as follows.
\begin{description}
\item[(A-B)] The propagator 
$\try{f}{\catchn{i_1}{g_1}{i_n}{g_n}}:X\to Y$ is 
defined from a catcher $H:X\to Y$ by:
$$ (\try{f}{\catchn{i_1}{g_1}{i_n}{g_n}})^\ppg \;=\; 
  (\toppg{ H})^\ppg : X \to Y$$
\begin{description}
\item[(1-2)] The catcher $H:X\to Y$ is defined 
from the propagator $f:X\to Y$ and from a catcher 
$k_1 = \catchn{i_1}{g_1}{i_n}{g_n}: \empt \to Y $ by:  
$$  H^\ctc \;=\; 
    \cotuple{ \id_Y^\pure \;|\; k_1^\ctc}^\ctc 
      \circ f^\ppg : 
    X\to Y $$
  $$ \xymatrix@C=3pc@R=1pc{
  & Y \ar[d]_{\id} \ar[rrrrd]^{\id} &&&& \\
  X \ar[r]^{f} & 
    Y \ar[rrrr]^(.4){h} &
    \ar@{}[ul]|(.4){\eqw} \ar@{}[dl]|(.4){\eqs} &&& Y \\ 
  & \empt \ar[u]^{\cotu} \ar[rrrru]_{k_1} &&&& \\
  } $$
\begin{description}
\item[(a-b)] The catcher $k_1: \empt \to Y$ 
is obtained by setting $p=1$ in the family of catchers 
$k_p = \catchn{i_p}{g_p}{i_n}{g_n} : \empt \to Y$ 
(for $p=1,\dots,n$) which are defined recursively by:   
 $$   k_p \;=\; \cotuple{g_p \;|\; k_{p+1}} \circ c_{i_p}
       \mbox{ for each } p=1,\dots,n 
			  \; \mbox{ and } \;  k_{n+1} = \cotu_Y   $$
 $$ \xymatrix@C=3pc@R=1pc{
  & P_{i_p} \ar[d]_{\id} \ar[rrrrd]^{g_p} &&&& \\
  \empt \ar[r]^(.4){c_{i_p}} & 
    P_{i_p}
      \ar[rrrr]^(.4){\cotuple{g_p|k_{p+1}}} &
    \ar@{}[ul]|(.4){\eqw} \ar@{}[dl]|(.4){\eqs} &&& Y \\ 
  & \empt \ar[u]^{\cotu} \ar[rrrru]_{k_{p+1}} &&&& \\
  }  $$
\end{description}
\end{description}
\end{description}
\end{definition}

It will be proved in Lemma~\ref{lem:coprod-cotu} 
that since $k_{n+1} = \cotu_Y$ we have $\cotuple{g_n | k_{n+1}} \eqs g_n$. 
It follows that when $n=1$ and 2 we get respectively: 
\begin{gather}
\label{eq:handle-deco-one}
\try{f}{\catch{i}{g}} \;\eqs\;  
   \bigtoppg{ 
    \left(\; \bigcotuple{ \id_Y \;|\; g \circ c_i } 
      \circ f \;\right) } \\
\label{eq:handle-deco-two}
\try{f}{\catchij{i}{g}{j}{h}} \;\eqs\;  
   \bigtoppg{
    \left(\;  
    \bigcotuple{\id \;|\; 
      \cotuple{g \;|\; 
        h \circ c_j} 
      \circ c_i} 
    \circ f\;\right) } 
\end{gather}

\subsection{Decorated models}
\label{subsec:deco-model}

Let $Exc$ be a set and $P_i$ (for $i\in I$) a family of sets
with injections $t_i:P_i\to \Exc$, such that 
$\Exc$ is the disjoint union of the images $E_i=t_i(P_i)$. 
Then we may define a decorated theory $\Tt_{\ex}$ as follows. 
A type is a set,
a pure expression $f^\pure:X\to Y$ is a function $f:X\to Y$,
a propagator $f^\ppg:X\to Y$ is a function $f:X\to Y+\Exc$ and 
a catcher $f^\ctc:X\to Y$ is a function $f:X+\Exc\to Y+\Exc$.  
For instance, each injection $t_i:P_i\to \Exc$ is a 
propagator $t_i^\ppg:P_i\to \emptyset$.
The conversion from pure expressions to propagators
is the construction of $f_1=\inn_Y\circ f_0:X\to Y+\Exc$
from $f_0:X\to Y$
and the conversion from propagators to catchers 
is the construction of $f_2=\cotuple{f\,|\,\ina_Y} :X+\Exc \to Y+\Exc$ 
from  $f_1:X\to Y+\Exc$. 
Composition of two expressions can be defined by
converting them to catchers and using the composition of functions
$f:X+\Exc \to Y+\Exc$ and $g:Y+\Exc \to Z+\Exc$ as 
$g\circ f:X+\Exc \to Z+\Exc$. 
When restricted to propagators this is compatible with the  
Kleisli composition with respect to the monad $X+\Exc$. 
When restricted to pure expressions this is compatible with the 
composition of functions $f_0:X \to Y$ and $g_0:Y \to Z$ as 
$g_0\circ f_0:X \to Z$. 
A strong equation $f^\ctc\eqs g^\ctc:X\to Y$ is an equality 
($\forall x\in X+\Exc, \; f(x)=g(x)$), and 
a weak equation $f\eqw g:X\to Y$ is an equality 
($\forall x\in X,\; f(x)=g(x)$); when restricted to 
propagators both notions coincide. 
The empty set $\emptyset$ is a decorated initial type 
and the family of propagators $(t_i^\ppg:P_i\to \emptyset)$ 
is a constitutive coproduct, because the 
family of functions $(t_i:P_i \to \Exc)$ is a coproduct 
in the category of sets. 

The \emph{models} of a decorated specification $\Ss$ 
with values in a decorated theory $\Tt$ are defined as kinds of 
morphisms from $\Ss$ to $\Tt$ in \cite{DD10-dialog}: 
a model maps each feature (type, pure expression, propagator, 
catcher, decorated initial type, constitutive coproduct, \dots) of $\Ss$ 
to a feature of the same kind in $\Tt$. 
When $\Tt$ is the theory $\Tt_{\ex}$ we recover 
the meaning of decorations as given informally 
in Section~\ref{subsec:deco-deco}. 
When in addition $\Ss$ is the specification $\Ss_{\ex}$ 
we get the following result.

\begin{thm}
\label{thm:deco-model} 
The decorated specification for exceptions $\Ss_{\ex}$ 
is sound with respect to the denotational semantics of 
Section~\ref{sec:exc},
in the sense that by mapping every feature in $\Ss_{\ex}$ 
to the feature with the same name in the decorated theory $\Tt_{\ex}$ 
we get a model of $\Ss_{\ex}$ with values in $\Tt_{\ex}$. 
\end{thm}

\begin{proof} 
For the tagging and untagging operations this is clear from the notations.
Then for the raising operations the result is obvious.
For the handling operations the result comes from  
a comparison of the steps (1-2) and (a-b) 
in Definitions~\ref{defi:deco-handle} and~\ref{defi:exc-handle},
while step (A-B) in Definition~\ref{defi:deco-handle} 
corresponds to the propagation of exceptions by the handling functions,
as in remark~\ref{rem:cases}. 
\end{proof} 

\section{Proofs involving exceptions}
\label{sec:proof}

As for proofs on states in \cite{DDFR12b-state}, 
we may consider two kinds of proofs on exceptions: 
the \emph{explicit} proofs involve a type of exceptions,
while the \emph{decorated} proofs do not mention any type of exceptions 
but require the specification to be decorated,
in the sense of Section~\ref{sec:deco}.
In addition, there is a simple procedure for deriving 
an explicit proof from a decorated one. 
In this Section we give some decorated proofs for exceptions, 
using the inference rules of Section~\ref{subsec:deco-deco}. 
Since the properties of the core tagging and untagging 
operations 
are dual to the properties of the looking up and updating operations 
we may reuse the decorated proofs involving states 
from \cite{DDFR12b-state}. 
Starting from any one of the seven equations for states
in \cite{PP02} we can dualize this equation 
and derive a property about raising and handling 
exceptions. This is done in this Section for two  
of these equations.

On states, the \emph{annihilation lookup-update} property  
means that updating any location with the content of this location
does not modify the state. 
A decorated proof of this property is given in \cite{DDFR12b-state}.
By duality we get the following \emph{annihilation untag-tag} property, 
which means that tagging just after untagging, 
both with respect to the same exceptional type,   
returns the given exception. 

\begin{lem}[Annihilation untag-tag]
\label{lem:ci-ti} 
For each $i\in I$:
$ 
 t_i^\ppg \circ c_i^\ctc \eqs \id_\empt^\pure \;.
$ 
\end{lem}

Lemma~\ref{lem:ci-ti} is used in Proposition~\ref{prop:hi-ri} 
for proving the \emph{annihilation catch-raise} property:
catching an exception by re-raising it is like doing nothing. 
First, let us prove Lemma~\ref{lem:coprod-cotu}, 
which has been used for getting Equation~(\ref{eq:handle-deco-one}).

\begin{lem}
\label{lem:coprod-cotu} 
For each propagator $g^\ppg:X\to Y$ we have 
$\cotuple{g\,|\,\cotu_Y}^\ctc \eqs g^\ppg$. 
\end{lem}

\begin{proof}
Since $\cotuple{g|\cotu_Y}$ is characterized up to strong equations 
by $\cotuple{g|\cotu_Y}\eqw g$ and 
$\cotuple{g|\cotu_Y}\circ \cotu_X \eqs \cotu_Y $,
we have to prove that $g\eqw g$ and $g\circ \cotu_X \eqs \cotu_Y $. 
The weak equation is due to the reflexivity of $\eqw$. 
The unicity of $\cotu_Y$ up to weak equations 
implies that $g\circ \cotu_X \eqw \cotu_Y$,
and since both members are propagators 
we get $g\circ \cotu_X \eqs \cotu_Y$. 
\end{proof}

\begin{prop}[Annihilation catch-raise]
\label{prop:hi-ri} 
For each propagator $f^\ppg:X\to Y$ and each $i\in I$:  
$ 
  \try{f}{\catch{i}{\throw{i}{Y}}} \eqs f   \;.
$ 
\end{prop}

\begin{proof} 
By Equation~(\ref{eq:handle-deco-one}) and 
Definition~\ref{defi:deco-raise} we have 
$ \try{f}{\catch{i}{\throw{i}{Y}}} \eqs  
   \toppg{ (\cotuple{ \id_Y | \cotu_Y \circ t_i \circ c_i } \circ f ) } $.
By Lemma~\ref{lem:ci-ti} 
$ \cotuple{ \id_Y | \cotu_Y \circ t_i \circ c_i } \eqs  
\cotuple{ \id_Y | \cotu_Y } $,
and the unicity property of $\cotuple{ \id_Y | \cotu_Y }$
implies that $\cotuple{ \id_Y | \cotu_Y } \eqs \id_Y $. 
Thus $ \try{f}{\catch{i}{\throw{i}{Y}}} \eqs \toppg{f}$.
Finally, since $\toppg{f} \eqw f$ and $f$ is a propagator 
we get $\toppg{f}\eqs f$.
\end{proof} 

On states, the \emph{commutation update-update} property  
means that updating two different locations can be done in any order. 
By duality we get the following \emph{commutation untag-untag} property, 
which means that untagging with respect to two distinct 
exceptional types can be done in any order. 
A detailed decorated proof of the 
commutation update-update property
is given in \cite{DDFR12b-state}.
The statement of this property and its proof 
use \emph{semi-pure products}, 
which were introduced in 
\cite{DDR11-seqprod} in order to provide a decorated alternative 
to the strength of a monad. 
Dually, the commutation untag-untag property 
use \emph{semi-pure coproducts},
which generalize the decorated coproducts $X\cong X+\empt$ 
from Definition~\ref{defi:deco-exc-coprod-ctc}. 
The \emph{coproduct} of two types $A$ and $B$ 
is defined as a type $A+B$ with two pure coprojections
$q_1^\pure:A \to A+B$ and $q_2^\pure:B \to A+B$, 
which satisfy the usual categorical coproduct property 
with respect to the pure morphisms. 
Then the \emph{semi-pure coproduct} of 
a propagator $f^\ppg:A\to C$
and a catcher $k^\ctc:B\to C$ is a catcher 
$\cotuple{f|k}^{\ctc}:A+B\to C$ 
which is characterized, up to strong equations, by 
the following decorated version of the coproduct property: 
$\cotuple{f|k} \circ q_1 \eqw f $
and $\cotuple{f|k} \circ q_2 \eqs k $.
Then as usual, the coproduct 
$f'+k':A+B\to C+D$ of 
a propagator $f':A\to C$
and a catcher $k':B\to D$ is the catcher 
$f'+k' = \cotuple{q_1\circ f\;|\; q_2\circ k}:A+B\to C+D$. 
Whenever $g$ is a propagator it can be proved that 
$\toppg{\cotuple{f|g}} \eqs \cotuple{f|g}$;  
thus, up to strong equation, 
we can assume that in this case 
$\cotuple{f\;|\;g}:A+B \to C$ is a propagator; 
it is characterized, up to strong equations, by
$\cotuple{f\;|\;g} \circ q_1 \eqs f$ and 
$\cotuple{f\;|\;g} \circ q_2 \eqs g$.

\begin{lem}[Commutation untag-untag] 
\label{lem:cj-ci} 
For each $i,j\in I$ with $i\ne j$: 
$$ 
  (c_i+ \id_{P_j})^\ctc \circ c_j^\ctc \eqs 
    (\id_{P_i} + c_j)^\ctc \circ c_i^\ctc  : \empt \to P_i+P_j
$$
\end{lem}

\begin{prop}[Commutation catch-catch] 
\label{prop:hj-hi}
For each $i,j\!\in\!I$ with $i\!\ne\! j$: 
$$ 
\try{f}{\catchij{i}{g}{j}{h}}\eqs \try{f}{\catchij{j}{h}{i}{g}} 
$$
\end{prop}

\begin{proof} 
According to Equation~(\ref{eq:handle-deco-two}): 
$  \try{f}{\catchij{i}{g}{j}{h}} \eqs 
  \toppg{( 
    \cotuple{\id \;|\; 
      \cotuple{g \;|\; 
        h \circ c_j} 
      \circ c_i} 
    \circ f )} $
Thus, the result will follow from
$ \cotuple{g \;|\; h \circ c_j} \circ c_i \eqs
  \cotuple{h \;|\;  g \circ c_i}  \circ c_j $. 
It is easy to check that 
 $ \cotuple{g \;|\; h \circ c_j}
  \eqs \cotuple{g \;|\; h} \circ (\id_{P_i} + c_j) $, 
so that
$ \cotuple{g \;|\; h \circ c_j} \circ c_i \eqs 
   \cotuple{g \;|\; h} \circ (\id_{P_i} + c_j) \circ c_i \;.$
Similarly
$ \cotuple{h \;|\; g \circ c_i} \circ c_j \eqs 
   \cotuple{h \;|\; g} \circ (\id_{P_j} + c_i) \circ c_j $
hence 
$ \cotuple{h \;|\; g \circ c_i} \circ c_j \eqs 
   \cotuple{g \;|\; h} \circ (c_i + \id_{P_j}) \circ c_j \;.$ 
Then the result follows from Lemma~\ref{lem:cj-ci}.
\end{proof} 



\begin{thebibliography}{99}

\bibitem{DD10-dialog} 
C\'esar Dom\'inguez, Dominique Duval. 
Diagrammatic logic applied to a parameterization process.  
Mathematical Structures in Computer Science 20, p.~639-654 (2010).

\bibitem{DDFR12a-short} 
Jean-Guillaume Dumas, Dominique Duval, Laurent Fousse, Jean-Claude Reynaud. 
A duality between exceptions and states.
Accepted for publication in MSCS.
arXiv:1112.2394.

\bibitem{DDFR12b-state} 
Jean-Guillaume Dumas, Dominique Duval, Laurent Fousse, Jean-Claude Reynaud.
Decorated proofs for computational effects: States.
Accepted for presentation at ACCAT'12. 
arXiv:1112.2396.

\bibitem{DDR11-seqprod} 
Jean-Guillaume Dumas, Dominique Duval, Jean-Claude Reynaud.
Cartesian effect categories are Freyd-categories.
Journal of Symbolic Computation 46, p.~272-293 (2011).

\bibitem{D03-diaspec} 
Dominique Duval. 
Diagrammatic Specifications. 
Mathematical Structures in Computer Science 13, p.~857-890 (2003).

\bibitem{java} 
James Gosling, Bill Joy, Guy Steele, Gilad Bracha. 
The Java Language Specification, Third Edition. 
Addison-Wesley Longman (2005).
\href{http://docs.oracle.com/javase/specs/jls/se5.0/jls3.pdf}{docs.oracle.com/javase/specs/jls/se5.0/jls3.pdf}. 

\bibitem{sml} 
Robert Harper. Programming in Standard ML. 
\href{http://www.cs.cmu.edu/~rwh/smlbook/book.pdf}{www.cs.cmu.edu/~rwh/smlbook/book.pdf}. 

\bibitem{haskell} 
The Haskell Programming Language. Monads. 
\href{http://www.haskell.org/haskellwiki/Monad}{www.haskell.org/haskellwiki/Monad}.

\bibitem{HP07} 
Martin Hyland, John Power.
The Category Theoretic Understanding of Universal Algebra: 
  Lawvere Theories and Monads. 
Electronic Notes in Theoretical Computer Science 172, p.~437-458 (2007).

\bibitem{Le06}
Paul Blain Levy.
Monads and adjunctions for global exceptions. 
MFPS 2006. 
Electronic Notes in Theoretical Computer Science 158, p.~261-287 (2006).

\bibitem{Mo91}  
Eugenio Moggi.
Notions of Computation and Monads.
Information and Computation 93(1), p.~55-92 (1991).

\bibitem{PP02} 
Gordon D. Plotkin, John Power.
Notions of Computation Determine Monads. 
FoSSaCS 2002.
Springer-Verlag 
Lecture Notes in Computer Science 2303, p.~342-356 (2002).

\bibitem{PP03} 
Gordon D. Plotkin, John Power.
Algebraic Operations and Generic Effects. 
Applied Categorical Structures 11(1), p.~69-94 (2003).

\bibitem{PP09} 
Gordon D. Plotkin, Matija Pretnar.
Handlers of Algebraic Effects. 
ESOP 2009.
Springer-Verlag 
Lecture Notes in Computer Science 5502, p.~80-94 (2009).

\bibitem{SM04} 
Lutz Schr{\"o}der, Till Mossakowski.
Generic Exception Handling and the Java Monad. 
AMAST 2004.
Springer-Verlag 
Lecture Notes in Computer Science 3116, p.~443-459 (2004).

\bibitem{Wa92} 
Philip Wadler.
The essence of functional programming.
POPL 1992.
ACM Press, p.~1-14 (1992).  

\end{thebibliography}
\end{document}